%% file: paper.tex
\documentclass[sigconf]{acmart}

\renewcommand\footnotetextcopyrightpermission[1]{} % removes footnote with conference information in first column
\usepackage{ifthen}
\usepackage[utf8]{inputenc}
\usepackage{latexsym, dsfont}
\usepackage{amsmath,amsthm}
\usepackage[inline]{enumitem}
\usepackage{siunitx}
\usepackage{graphicx}
\usepackage{caption}
\usepackage[capitalize]{cleveref}
\usepackage{todonotes}
\usepackage[caption=false,font=footnotesize]{subfig}
\usepackage{mathtools}
\usepackage{tikz}
\newboolean{PGFPlots}
\setboolean{PGFPlots}{true}
\newcommand{\IR}{\mathds{R}}
\newtheorem{thm}{Theorem}
\newtheorem{lem}[thm]{Lemma}

\usetikzlibrary{arrows,shapes,decorations,decorations.pathreplacing,decorations.pathmorphing,intersections,calc,fit}

\ifthenelse{\boolean{PGFPlots}}{
\usepackage{pgfplots}
\usepgfplotslibrary{units,groupplots}
\usetikzlibrary{external}
\tikzexternalize
\tikzset{external/only named=true}
}{}

\usepackage[appendix=append]{apxproof} 
\usepackage[append]{optional}

\newtheoremrep{theorem}{Theorem}[section]

\input{common.tex}
\input{defstikz}

\input{tikzfigures}

\newboolean{conference}
\setboolean{conference}{true}  

\crefname{equation}{}{}

\author{Arman Ferdowsi}
\affiliation{%
  \institution{TU Wien}
  \city{Vienna}
  \country{Austria}}
\email{aferdowsi@ecs.tuwien.ac.at}

\author{Matthias Függer}
\affiliation{%
  \institution{CNRS \& LMF, ENS Paris-Saclay, Université Paris-Saclay \& Inria}
  \city{Gif-sur-Yvette}
  \country{France}}
\email{mfuegger@lsv.fr}

\author{Thomas Nowak}
\affiliation{%
  \institution{Université Paris-Saclay, CNRS, ENS Paris-Saclay}
  \institution{Institut Universitaire de France}
  \city{Gif-sur-Yvette}
  \country{France}
  }
\email{thomas@thomasnowak.net}

\author{Ulrich Schmid}
\affiliation{%
  \institution{TU Wien}
  \city{Vienna}
  \country{Austria}}
\email{s@ecs.tuwien.ac.at}

\copyrightyear{2023} 
\acmYear{2023} 
\setcopyright{rightsretained} 
\acmConference[HSCC '23]{Proceedings of the 26th ACM International Conference on Hybrid Systems: Computation and Control}{May 9--12, 2023}{San Antonio, TX, USA}
\acmBooktitle{Proceedings of the 26th ACM International Conference on Hybrid Systems: Computation and Control (HSCC '23), May 9--12, 2023, San Antonio, TX, USA}
\acmDOI{10.1145/3575870.3587125}
\acmISBN{979-8-4007-0033-0/23/05}


\title{Continuity of Thresholded Mode-Switched ODEs and Digital Circuit Delay Models}
\thanks{This research was supported by the Austrian Science Fund (FWF) project DMAC (grant no. ~\mbox{P32431})
  and the ANR project DREAMY (ANR-21-CE48-0003).}

\begin{document}

\begin{abstract}
Thresholded mode-switched ODEs are restricted dynamical systems that switch ODEs depending on digital input signals only, and produce a digital output signal by
thresholding some internal signal.
Such systems arise in recent digital circuit delay models, where the
analog signals within a gate are governed by ODEs that change depending on the digital inputs.

We prove the continuity of the mapping from digital input signals to digital output signals for a large class of thresholded mode-switched ODEs. This continuity property is known to be instrumental for ensuring the faithfulness of the model w.r.t.\
propagating short pulses. We apply our result to several instances of such digital delay models, thereby proving them to be faithful.
\end{abstract}

	\maketitle	
	
	\section{Introduction}
	\label{sec:intro}

A natural class of hybrid systems can be described by
  the dynamics of a continuous process, which is controlled by externally supplied digital mode switch signals, and provides a digital output based
  on whether some internal signal crosses a threshold, see \cref{fig:switched}
  for an illustration.   
Examples are digitally controlled thermodynamic processes,
  hydrodynamic systems, and, in particular, digital integrated circuits.
The continuous dynamics of these systems are described by
  \emph{Ordinary Differential Equations} (ODEs)
  for the temperature, the pipe's pressures and fill-levels, 
  or the gate's currents and voltages over time.
Digital mode switches are used to switch between ODE systems, e.g., by
  turning on a heater, closing a valve, or applying an input transition
  to a gate's input.
The environment of the hybrid system is only notified if the temperature
  or fill-level crosses a threshold, or, in the case of a digital gate,
  is said to produce an output transition when some internal
  voltage crosses a threshold. 

\begin{figure}[tb]
  \centerline{
    \includegraphics[width=0.85\linewidth]{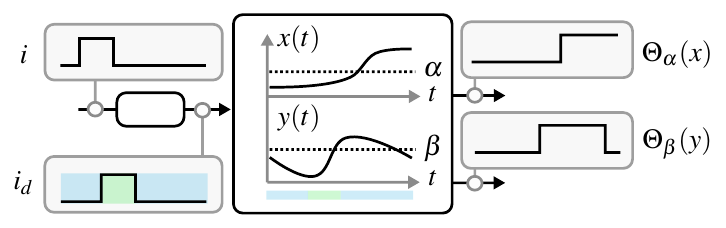}
  }
  \caption{Thresholded mode-switched ODE with a single mode input $i$, the delayed
    input $i_d$, two continuous
    states $x,y$, and two thresholded outputs $\Theta_\alpha(x)$ and $\Theta_\beta(y)$.}
  \label{fig:switched}
\end{figure}

In this work, we consider the composition of such hybrid systems in a
circuit, where digital threshold signals of one component drive mode switch
  signals of a downstream component.
We give conditions that ensure the continuity of the outputs of such circuits with respect to their inputs
  and provide two application examples in the context of circuit delay models.
The proven continuity property shows that small variations of the inputs lead to
  small variations of the output signal, a property that is necessary
  for digital circuit models to be consistent with physical analog ODE models.

\smallskip\noindent\textbf{Digital circuits, continuity, and faithful delay models.}
Analog simulations of digital circuits are time-consuming and
  are thus replaced by digital simulations whenever possible.
Typical application domains that require simulation of precise circuit transition times are
  particularly timing-critical, asynchronous parts of a circuit, e.g., inter-neuron links
  using time-based encoding in hardware-implemented spiking neural networks~\cite{BVMRRVB19},
  where the worst-case delay estimates provided by static timing analysis techniques are
  not sufficient for ensuring correct operation.

A mandatory prerequisite for dynamic timing analysis are
  digital delay models, which allow to accurately determine the input-to-output
  delay of every constituent gate in a circuit.
Suitable models must also account 
  for the fact that the delay of an individual
  signal transition usually depends on the previous transition(s), in particular, when they 
  were close.
The simplest class of such models are \emph{single-history delay models} \cite{FNS16:ToC,BJV06,FNNS19:TCAD}, where the 
  input-to-output delay $\delta(T)$ of a gate
  depends on the previous-output-to-input delay~$T$.

It has been proved by F\"ugger et al.~\cite{FNNS19:TCAD} that a certain continuity 
  property of single-history models is mandatory for the digital 
  abstraction to faithfully model the analog reality.
In particular, the predicted output
  transitions must not be substantially affected by arbitrarily short input glitches.
For example, the constant-low input signal and an arbitrarily short low-high-low pulse must
  produce arbitrary close gate output signals.
So far, the only delay model known to ensure this continuity property is the \emph{involution delay model} (IDM) \cite{FNNS19:TCAD}, 
  which consists of zero-time Boolean gates interconnected by single-input
  single-output involution delay channels.
An IDM channel is characterized by a delay function~$\delta$,
  which is a negative involution, i.e., $-\delta(-\delta(T))=T$.
In its generalized version, different delay functions $\dup$ resp.\ $\ddo$
  are assumed for rising resp.\ falling transitions, requiring 
  $-\dup(-\ddo(T))=T$.
Unlike all other existing delay models,
  the IDM has been proved to faithfully model 
  glitch propagation for the so-called short-pulse filtration problem~\cite{FNNS19:TCAD}, 
  and is hence the only candidate for a faithful delay model known so far~\cite{FNS16:ToC}.

It has also been shown~\cite{FNNS19:TCAD} that involution delay functions 
  arise naturally in a 2-state thresholded hybrid channel model, which
  consists of a pure delay component, a slew-rate limiter with a rising and falling switching
  waveform, and an ideal comparator (\cref{fig:analogy}):
The binary-valued input~$i_a$ is delayed by $\dmin>0$, which assures causality of
  channels, i.e., $\delta_{\uparrow/\downarrow}(0)>0$.
For every transition on~$i_d$, the generalized slew-rate limiter switches to the corresponding waveform
  ($\fdo/\fup$ for a falling/rising transition).
The essential property here is that the analog output voltage~$o_a$ is a \emph{continuous} (but not necessarily smooth) function of time.
Finally, the comparator generates the output~$o_d$ by
  digitizing~$o_a$ w.r.t.\ the discretization threshold voltage~$\vth$.

\begin{figure}[tb]
  \centerline{
    \includegraphics[width=0.55\linewidth]{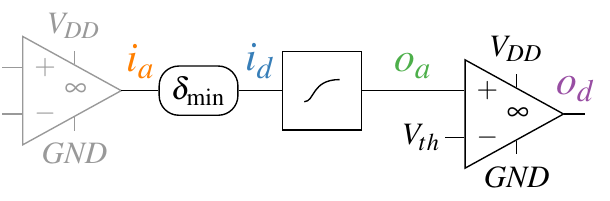}}
  \centerline{
    \includegraphics[width=0.8\linewidth]{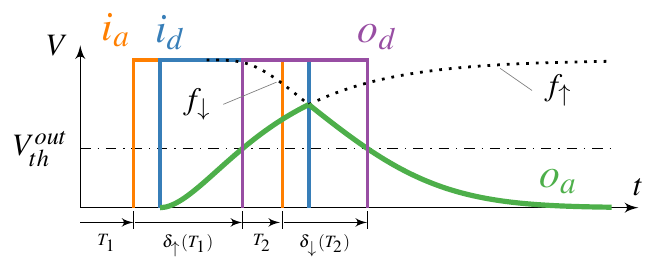}}\vspace{-0.3cm}
  \caption{Hybrid involution delay channel model (upper part) with a sample execution
    (bottom part). Adapted from~\cite{FNNS19:TCAD}.}\vspace{-0.4cm}
  \label{fig:analogy}
\end{figure}

\medskip

Whereas the accuracy of IDM predictions for single-input, single-output circuits 
  like inverter chains or clock trees turned out to be very good, this is less so for circuits 
  involving multi-input gates~\cite{OMFS20:INTEGRATION}.
It has been revealed by Ferdowsi et al.~\cite{FMOS22:DATE} that this is primarily due to the 
  IDM's inherent lack of properly covering
  output delay variations caused by \emph{multiple input switching} (MIS) in close
  temporal proximity~\cite{CGB01:DAC}, also known as the \emph{Charlie effect}:
compared to the \emph{single input switching}
  case, output transitions are sped up/slowed down with decreasing
  transition separation time on different inputs.
Single-input,
  single-output delay channels like IDM cannot exhibit such a behavior.

To capture MIS effects in a 2-input \NOR\ gate, Ferdowsi et al.~\cite{FMOS22:DATE} 
  hence proposed an alternative digital delay model based on a 4-state hybrid 
  gate model.
It has been obtained by replacing the 4 transistors in
  the RC-model of a CMOS \NOR\ gate by ideal zero-time switches, which
  results in one mode per possible digital state of the inputs 
  $(A,B)\in \{(0,0), (0,1), (1,0), (1,1)\}$. 
In each mode, the voltage of the the output signal and an internal node
  are governed by constant-coefficient first order ODEs.
When an input signal changes its state, the system
  switches to the new mode and its corresponding ODEs.

Albeit digitizing this hybrid gate model, using a comparator with a suitable
  threshold voltage $\vth$ as in \cref{fig:analogy}, leads to a 
  quite accurate digital delay model, it turned out to still fail to 
  capture the MIS delay for a rising output transition.
In a follow-up paper \cite{ferdowsi2022accurate}, Ferdowsi et.~al.\ hence introduced a 
  refined gate delay model, where the switching-on of the pMOS transistors is not instantaneous,
  but rather governed by a simple time evolution function $\sim 1/t$, inspired
  by the Shichman-Hodges transistor model~\cite{ShichmanHodges}.
The resulting 4-state hybrid model consists of a single not-constant-coefficient first-order ODE 
  per mode, and has been shown to accurately model MIS effects.

Whereas the experimental evaluation of the modeling accuracy of the
  hybrid models discussed above shows that they outperform the simple 
  IDM model \cite{OMFS20:INTEGRATION}, it is not clear whether they are
  also \emph{faithful} digital delay models.
What would be needed here is a proof that the digital delay models 
  obtained by digitizing hybrid models
  satisfy the continuity property required for faithfulness.

\smallskip\noindent\textbf{Contributions.}
Our paper answers this question in the affirmative.
More generally, we prove that any thresholded hybrid model like the one shown in
  \cref{fig:switched} that satisfies some mild 
  conditions on their ODEs results in a continuous digitized hybrid model.
We then show that the above hybrid gate models
  fall into this category, and that the proven continuity implies faithful short-pulse
  propagation of any such model. Since the square of a signal is (proportional to) 
  its power, this also implies a continuity property from the input signal power 
to the output signal power. Consequently, these delay models are indeed promising 
candidates for the correct and timing+power-accurate 
simulation of digital circuits.
In more detail:
\begin{enumerate}
\item[(1)] We show that any hybrid model, where mode $m$ 
  is governed by a system of first-order ODEs $\frac{dx}{dt} = f_m(t,x)$,
  leads to a continuous digital delay model, provided all the $f_m$ are 
  continuous in $t$ and Lipschitz continuous in $x$, with a common 
  Lipschitz constant for every $t>0$ and $m$.

\item[(2)] We prove that the parallel composition of finitely many digitized hybrid gates
in a circuit result in a unique and Zeno-free execution, under some mild conditions
regarding causality. In conjunction with our
continuity result, we prove that the resulting model is faithful w.r.t.\ solving
the canonical short-pulse filtration problem.

\item[(3)] We demonstrate that the hybrid gate models proposed in 
  \cite{FMOS22:DATE, ferdowsi2022accurate} satisfy these properties, and are
  hence continuous and thus faithful. 
\end{enumerate}

\smallskip
\noindent\textbf{Paper organization.}
In \cref{sec:continuity}, we instantiate our general continuity result (\cref{thm:thr:cont}).
\cref{sec:gatemodels} presents our main continuity result for hybrid gate models (\cref{thm:channels:are:cont} and \cref{thm:digitizedmodels:are:cont}), and
\cref{sec:circuits} deals with circuit composition.
In \cref{sec:examples}, we provide
  examples for the hybrid models considered in this work:
  a simple heater from literature \cite{henzinger2000theory}, the simple
  hybrid gate model \cite{FMOS22:DATE}, and the advanced gate model \cite{ferdowsi2022accurate}.
Some conclusions and directions of future research are provided 
  in \cref{sec:conclusions}.

\section{Thresholded Mode-Switched ODEs}
\label{sec:continuity}

In this section, we provide a generic proof that every hybrid model that adheres to some mild
  conditions on its ODEs leads to a continuous digital delay model.
We start with proving continuity in the analog domain and then establish continuity
  of the digitized signal obtained by feeding a continuous real-valued
  signal into a threshold voltage comparator.
Combining those results will allow us to assert the continuity of digital delay channels
  like the one shown in \cref{fig:analogy}.

\subsection{Continuity of ODE mode switching}
\label{sec:contODE}

For a vector $x\in\IR^n$, denote by~$\lVert x\rVert$ its Euclidean norm.
For a piecewise continuous function $f:[a,b]\to\IR^n$, we write $\lVert f\rVert_1 = \int_a^b \lVert f(t) \rVert \, dt$ for its $1$-norm and
$\lVert f\rVert_\infty = \sup_{t\in[a,b]} \lVert f(t)\rVert$ for its supremum norm.
The projection function of a vector in $\IR^n$ onto its $k^\text{th}$ component, for $1 \leq k \leq n$,
is denoted by $\pi_k : \IR^n \to \IR$. 

In this section, we will consider non-autonomous first-order ODEs of the form 
$\frac{d}{dt}\,x(t) = f(t,x(t))$, where the non-negative $t \in \IR_+$ represents
the time parameter, $x(t) \in U $ for some arbitrary open set $U\subseteq \IR^n$,
$x_0 \in U$ is some initial value, and $f:\IR_+ \times U\to\IR^n$ is chosen
from a set $F$ of bounded functions that are continuous for $(t,x) \in [0,T] \times U$,
where $0 < T < \infty$, and Lipschitz continuous in $U$ with 
a common Lipschitz constant for all $t\in[0,T]$ and all choices 
of $f \in F$. It is well-known that every such ODE has a unique solution $x(t)$ 
with $x(0)=x_0$ that satisfies $x(t)\in U$ for $t\in[0,T]$, is continuous in $[0,T]$, 
and differentiable in $(0,T)$.

The following lemma shows the continuous dependence of the solutions of such ODEs 
on their initial values. To be more explicit, the exponential dependence of the Lipschitz constant on the time parameter allows temporal composition of the bound.
The proof can be found in standard textbooks on ODEs~\cite[Theorem~2.8]{teschl2012ordinary}.

\begin{lem}\label{lem:cont:wrt:initial:value}
Let $U\subseteq \IR^n$ be an open set and let $f:\IR \times U\to\IR^n$ be Lipschitz continuous with Lipschitz constant~$K$ for $t\in[0,T]$ with $T > 0$, and let $x,y:[0,T]\to U$ be continuous functions that are differentiable on $(0,T)$ such that $\frac{d}{dt}\,x(t) = f(t,x(t))$
and $\frac{d}{dt}\,y(t) = f(t,y(t))$
for all $t\in (0,T)$. Then, $\lVert x(t) - y(t) \rVert
\leq
e^{t K} \lVert x(0) - y(0) \rVert$ for all $t\in[0,T]$.
\end{lem}

A \emph{step function} $s:\IR_+ \to \{0,1\}$ is a right-continuous function 
with left limits, i.e., $\lim_{t\to t_0^+} s(t)=s(t_0)$ and $\lim_{t\to t_0^-} s(t)$
exists for all $t_0 \in \IR_+$.
A \emph{binary signal} $s$ is a step function $s:[0,T]\to\{0,1\}$,
a \emph{mode-switch signal} $a$ is a step function $a:[0,T]\to F$, $t\mapsto a_t$.

Given a mode-switch signal~$a$, a \emph{matching output signal} for~$a$ is 
a function $x_a:[0,T]\to U$ that satisfies
\begin{enumerate}
\item[(i)] $x_a(0) = x_0$,
\item[(ii)] the function~$x_a$ is continuous,
\item[(iii)] for all $t\in(0,T)$, if~$a$ is continuous at~$t$, then~$x_a$ is differentiable at~$t$ and $\frac{d}{dt}\,x_a(t) = a_t(t,x_a(t))$.
\end{enumerate}
For (iii), recall that the domain of $a$ is $F$.

\begin{lemma}[Existence and uniqueness of matching output signal] \label{lem:uniqueness}
Given a mode-switch signal~$a$, the matching output signal $x_a$ for~$a$
exists and is unique.
\end{lemma}

\begin{proof}
$x_a$ can be constructed inductively, by pasting together the solutions 
$x_{t_j}$ of $\frac{d}{dt}\,x_{t_j}(t) = a_{t_j}(t,x_{t_j}(t))$, where $t_0=0$
and $t_1 < t_2 < \dots$ are $a$'s switching times in $S_a$: 
For the induction basis $j=0$, we define $x_a(t):=x_{t_0}(t)$
with initial value $x_{t_0}=x_{t_0}(t_0):=x_0$ for $t\in [0,t_1]$. Obviously, (i) holds
by construction, and the continuity and differentiability of $x_{t_0}(t)$ at other times
ensures (ii) and (iii).

For the induction step $j \to j+1$, we assume 
that we have constructed $x_a(t)$ already for $0 \leq t \leq t_j$. For $t\in [t_j,t_{j+1}]$, we 
define $x_a(t):=x_{t_{j+1}}(t)$ with initial value $x_{t_{j+1}}=x_{t_{j+1}}:=x_a(t_j)=x_{t_j}(t_j)$.
Continuity of $x_a(t)$ at $t=t_j$ follows by construction, and the continuity and differentiability 
of $x_{t_{j+1}}(t)$ again ensures~(ii) and~(iii).
\end{proof}

Given two mode-switch signals~$a$, $b$, we define their distance as
\begin{equation}
d_T(a,b)
=
\lambda\big(\{ t\in[0,T] \mid a_t \neq b_t \}\big)
\end{equation}
where~$\lambda$ is the Lebesgue measure on~$\IR$.
The distance function~$d_T$ is a metric on the set of mode-switch signals.

The following \cref{thm:real-valued:is:cont} shows that the mapping 
$a\mapsto x_a$ is continuous.

\begin{thm}\label{thm:real-valued:is:cont}
Let~$K\geq 1$ be a common Lipschitz constant for all functions in~$F$ and let~$M$ be a real number such that $\lVert f(t,x(t))\rVert \leq M$ for all $f\in F$, all $x\in U$, and all $t \in [0,T]$.
Then, for all mode-switch signals~$a$ and~$b$,
if~$x_a$ is the output signal for~$a$
and~$x_b$ is the output signal for~$b$,
then
$\lVert x_a - x_b\rVert_{\infty} \leq 2Me^{TK} d_T(a,b)$. Consequently, the mapping $a\mapsto x_a$ is continuous.
\end{thm}
\begin{proof}
Let $S = \{ t\in (0,T) \mid \text{$a$ or $b$ is discontinuous at $t$} \} \cup \{0,T\}$ be the set of switching times of~$a$ and~$b$.
The set~$S$ must be finite, since both~$a$ and~$b$ are right-continuous on a compact interval.
Let $0=s_0 < s_1 < s_2 < \cdots < s_m=T$ be the increasing enumeration of~$S$.

We show by induction on~$k$ that
\begin{equation}\label{eq:thm:real-valued:is:cont:induction}
\forall t\in [0, s_k] \colon\quad
\lVert x_a(t) - x_b(t) \rVert \leq 2Me^{t K} d_t(a,b)
\end{equation}
for all $k\in\{0,1,2,\dots,m\}$.
The base case $k=0$ is trivial.
For the induction step $k \mapsto k+1$, we distinguish the two cases
$a_{s_k} = b_{s_k}$ and
$a_{s_k} \neq b_{s_k}$.

If $a_{s_k} = b_{s_k}$, then we have $a_t = b_t$ for all $t\in [s_k,s_{k+1})$ and
hence $d_t(a,b) = d_{s_k}(a,b)$ for all $t\in[s_k,s_{k+1}]$.
Moreover, we can apply Lemma~\ref{lem:cont:wrt:initial:value} and obtain
\begin{equation}
\forall t\in [s_k, s_{k+1}]\colon\quad
\lVert x_a(t) - x_b(t) \rVert
\leq
e^{(t - s_k)K} \lVert x_a(s_k) - x_b(s_k) \rVert
\enspace.
\end{equation}
Plugging in~\eqref{eq:thm:real-valued:is:cont:induction} for $t = s_k$ reveals that~\eqref{eq:thm:real-valued:is:cont:induction} holds for all $t\in[s_k, s_{k+1}]$ as well.

If $a_{s_k} \neq b_{s_k}$, then~$x_a$ and~$x_b$ follow different differential equations in the interval $t \in [s_k,s_{k+1}]$.
We can, however, use the mean-value theorem for vector-valued functions~\cite[Theorem~5.19]{rudin1976principles} to obtain
\begin{equation}
\forall t\in [s_k, s_{k+1}]\colon\quad
\lVert x_a(t) - x_a(s_k) \rVert
\leq
M(t - s_k)\enspace\text{and}
\end{equation}
\begin{equation}
\forall t\in [s_k, s_{k+1}]\colon\quad
\lVert x_b(t) - x_b(s_k) \rVert
\leq
M(t - s_k).
\end{equation}
This, combined with the induction hypothesis, the equality $d_t(a,b) = d_{s_k}(a,b) + (t-s_k)$, and the inequalities $1\leq e^{t K}$ and $e^{s_k K} \leq e^{t K}$, implies
\begin{equation*}
\begin{split}
\lVert x_a(t) - x_b(t) \rVert
& \leq
\lVert x_a(t) - x_a(s_k) \rVert\\
& \quad + \lVert x_a(s_k) - x_b(s_k) \rVert
+
\lVert x_b(s_k) - x_b(t) \rVert
\\ & \leq
2M(t-s_k)
+
2Me^{s_k K} d_{s_k}(a,b)
\\ & \leq
2M e^{t K} (t-s_k)
+
2Me^{t K} d_{s_k}(a,b)
\\ & =
2M e^{t K} \big( d_t(a,b) - d_{s_k}(a,b) \big)
+
2Me^{t K} d_{s_k}(a,b)
\\ & =
2M e^{t K} d_t(a,b)
\end{split}
\end{equation*}
for all $t\in [s_k, s_{k+1}]$.
This concludes the proof.
\end{proof}

We conclude this section with the remark that the (proof of the)
continuity property
of \cref{thm:real-valued:is:cont} is very different from the standard
(proof of the) continuity property of controlled variables in
closed thresholded hybrid systems. Mode switches in such systems
are caused by the time evolution of the system
itself, e.g., when some controlled variable exceeds some value. Consequently,
such systems can be described by means of a \emph{single} ODE system with
discontinuous righthand side \cite{Fil88}.

By contrast, in our hybrid systems, the mode switches are solely caused by changes
of digital inputs that are \emph{externally} controlled: For every possible pattern of
the digital inputs, there is a dedicated ODE system that controls the 
analog output. Consequently, the time evolution of the output now also depends 
on the time evolution of the inputs. Proving the
continuity of the (discretized) output w.r.t.\ different 
(but close, w.r.t. some metric) digital input signals requires
relating the output of \emph{different} ODE systems.

\subsection{Continuity of thresholding}
\label{sec:thresholding}

For a real number $\xi\in\IR$ and a function $x:[a,b]\to\IR$, denote by $\Theta_\xi(x)$ the thresholded version of $x$ defined by
\begin{equation}
\Theta_\xi(x) : [a,b] \to \{0,1\},
\quad
\Theta_\xi(x)(t) =
\begin{cases}
0 & \text{if } x(t) \leq \xi,\\
1 & \text{if } x(t) > \xi.
\end{cases}
\end{equation}

\begin{lem}\label{lem:thr:cont:monotonic:end}
Let $\xi\in\IR$ and let $x:[a,b]\to\IR$ be a continuous strictly monotonic function with $x(b)=\xi$.
Then, for every $\varepsilon>0$, there exists a $\delta > 0$ such that,
for every continuous function $y:[a,b]\to\IR$, the condition
$\lVert x - y\rVert_\infty < \delta$ implies $\lVert \Theta_\xi(x) - \Theta_\xi(y)\rVert_1 < \varepsilon$.
\end{lem}
\begin{proof}
We show the lemma for the case that~$x$ is strictly increasing.
The proof for strictly decreasing~$x$ is analogous.

Set $\chi = x(a)$.
Since~$x$ is bijective onto the interval $[\chi,\xi]$, it has an inverse function $x^{-1}:[\chi,\xi]\to [a,b]$.
The inverse function~$x^{-1}$ is continuous because the domain $[a,b]$ is compact \cite[Theorem~4.17]{rudin1976principles}.

The relation $t \leq x^{-1}(\xi - \delta)$ implies $x(t)+\delta \leq \xi$.
Hence, if $\lVert x - y\rVert_\infty < \delta$, then $y(t) \leq x(t) + \delta \leq \xi$ for all $t\leq x^{-1}(\xi-\delta)$.
This means that $\Theta_\xi(y)(t) = 0$ for all $t\leq x^{-1}(\xi-\delta)$,
so $t>x^{-1}(\xi-\delta)$ for every $t\in[a,b]$ where $\Theta_\xi(y)(t) = 1$.

By assumption, we have $\Theta_\xi(x)(t) = 0$ for all $t\in [a,b]$.
Thus,
\begin{equation}
\begin{split}
\lVert \Theta_\xi(x) - \Theta_\xi(y)\rVert_1
& =
\lambda\big( \{ t\in[0,T] \mid \Theta_\xi(y) = 1 \} \big)
\\ & =
\lambda\big( \{ t\in[0,T] \mid y(t) > \xi \} \big) \label{eq:measure1}
\\ & \leq
b - x^{-1}(\xi-\delta).
\end{split}
\end{equation}
Note that continuity of $y$ is sufficient to ensure that the set in \cref{eq:measure1} is measurable.
Since~$x^{-1}$ is continuous, we have $x^{-1}(\xi-\delta) \to x^{-1}(\xi) = b$ as $\delta \to 0$.
In particular, for every $\varepsilon > 0$, there exists a $\delta > 0$ such that $b - x^{-1}(\xi-\delta) < \varepsilon$.
This concludes the proof.
\end{proof}

The following \cref{lem:thr:cont:monotonic} shows that we can drop the assumption $x(b)=\xi$
in \cref{lem:thr:cont:monotonic:end}:

\begin{lem}\label{lem:thr:cont:monotonic}
Let $\xi\in\IR$ and let $x,y:[a,b]\to\IR$ be two continuous functions where $x$ is strictly monotonic. Then, for every $\varepsilon>0$, there exists a $\delta > 0$ such that $\lVert x - y\rVert_\infty < \delta$ implies $\lVert \Theta_\xi(x) - \Theta_\xi(y)\rVert_1 < \varepsilon$.
\end{lem}
\begin{proof}
We again show the lemma for the case that~$x$ is strictly increasing.
The proof for strictly decreasing~$x$ is analogous.

Let $\varepsilon > 0$. We distinguish three cases:

(i) If $x(b) < \xi$, then we have $\Theta_\xi(x)(t)=0$ for all $t\in[a,b]$.
Choosing $\delta = \xi - x(b)$, we deduce $y(t) < x(t) + \delta \leq x(b) + \xi - x(b) = \xi$ for all $t\in[a,b]$ whenever $\lVert x - y\rVert_\infty < \delta$.
Hence, we get $\Theta_\xi(y)(t)=0$ for all $t\in[a,b]$ and thus
$\lVert \Theta_\xi(x) - \Theta_\xi(y)\rVert_1 = 0 < \varepsilon$.

(ii) If $x(a) > \xi$, then we can choose $\delta = x(a) - \xi$ and get $\Theta_\xi(y)(t) = \Theta_\xi(x)(t) = 1$ for all $t\in[a,b]$ whenever $\lVert x-y \rVert_\infty < \delta$.
In particular, $\lVert \Theta_\xi(x) - \Theta_\xi(y)\rVert_1 = 0 < \varepsilon$.

(iii) If $x(a) \leq \xi \leq x(b)$, then there exists a unique $c\in[a,b]$ with $x(c) = \xi$.
Applying Lemma~\ref{lem:thr:cont:monotonic:end} on the restriction of~$x$ on the interval $[a,c]$, we get the existence of a $\delta_1>0$ such that $\lVert x- y\rVert_{[a,c],\infty} < \delta_1$ implies $\lVert \Theta_\xi(x) - \Theta_\xi(y)\rVert_{[a,c],1} < \varepsilon/2$; herein, $\lVert \cdot \rVert_{[a,c],\infty}$ and $\lVert \cdot \rVert_{[a,c],1}$ denote the supremum-norm and the $1$-norm on the interval~$[a,c]$, respectively.
Applying Lemma~\ref{lem:thr:cont:monotonic:end} on the restriction of~$x$ on the interval~$[c,b]$ after the coordinate transformation $t\mapsto -t$ yields the existence of a $\delta_2>0$ such that $\lVert x- y\rVert_{[c,b],\infty} < \delta_2$ implies $\lVert \Theta_\xi(x) - \Theta_\xi(y)\rVert_{[c,b],1} < \varepsilon/2$.
Setting $\delta = \min\{\delta_1,\delta_2\}$, we thus get
$\lVert \Theta_\xi(x) - \Theta_\xi(y)\rVert_{[a,b],1}
 =
\lVert \Theta_\xi(x) - \Theta_\xi(y)\rVert_{[a,c],1}
+
\lVert \Theta_\xi(x) - \Theta_\xi(y)\rVert_{[c,b],1}
 <
\varepsilon/2 + \varepsilon/2
=
\varepsilon
$
whenever $\lVert x - y\rVert_{[a,b],\infty} < \delta$.
\end{proof}

The following \cref{thm:thr:cont} shows that the mapping $x\mapsto \Theta_\xi(x)$ 
is continuous for a given function $x$, provided that $x$ has only finitely many local optima, i.e.,
points where $x'(t)=0$:

\begin{thm}\label{thm:thr:cont}
Let $\xi\in\IR$ and let $x,y:[0,T]\to\IR$ be two differentiable functions.
Assume that~$x$ has only finitely many local optima.
Then, for every $\varepsilon>0$, there exists a $\delta > 0$ such that $\lVert x - y\rVert_\infty < \delta$ implies $\lVert \Theta_\xi(x) - \Theta_\xi(y)\rVert_1 < \varepsilon$. Consequently, the mapping $x\mapsto \Theta_\xi(x)$ is continuous.
\end{thm}
\begin{proof}
Let $\mathcal{N} = \left\{t\in[0,T] \mid \text{$x$ has a local optimum at $t$}\right\}\cup\{0,T\}$, which is
finite by assumption, and $t_0 < t_1 < t_2 < \cdots < t_m$ be the increasing enumeration of~$\mathcal{N}$.
By the mean-value theorem, the function~$x$ is strictly monotonic in every interval $[t_k,t_{k+1}]$ for $k\in\{0,1,2,\dots,m-1\}$.

Let $\varepsilon>0$.
Applying Lemma~\ref{lem:thr:cont:monotonic} to the restriction of~$x$ on each of the intervals $[t_k,t_{k+1}]$, we get the existence of $\delta_k>0$ such that $\lVert x - y\rVert_{[t_k,t_{k+1}],\infty} < \delta_k$ implies
$\lVert \Theta_\xi(x) - \Theta_\xi(y)\rVert_{[t_k,t_{k+1}],1} < \varepsilon/m$ for each $k\in\{0,1,2,\dots,m-1\}$.
Setting $\delta = \min\{\delta_0, \delta_1, \delta_2, \dots, \delta_{m-1}\}$, we thus obtain
\begin{equation}
\begin{split}
\lVert \Theta_\xi(x) - \Theta_\xi(y)\rVert_{[0,T],1}
& =
\sum_{k=0}^{m-1}
\lVert \Theta_\xi(x) - \Theta_\xi(y)\rVert_{[t_k,t_{k+1}],1}
\\ & <
\sum_{k=0}^{m-1}
\varepsilon/m
=
\varepsilon
\end{split}
\end{equation}
whenever $\lVert x - y\rVert_{[0,T],\infty} < \delta$.
\end{proof}

\section{Continuity of Digitized Hybrid Gate}
\label{sec:gatemodels} 

To prepare for our general result about the continuity of hybrid gate models, we
will first
(re)prove the continuity of IDM channels as shown in \cref{fig:analogy}, which
has been established by a quite tedious direct proof in \cite{FNNS19:TCAD}. 
In our notation, an IDM channel consists of:
\begin{itemize}
\item A nonnegative minimum delay $\dmin\geq0$ and a delay function $\Delta_{\dmin}(s)$ that maps the binary input signal 
$i_a$, augmented with the left-sided limit $i_a(0-)$ as the \emph{initial value}\footnote{In \cite{FNNS19:TCAD}, this initial
value of a signal was encoded by extending the time domain to the whole $\IR$ and using $i_a(-\infty)$.} that can be different from $i_a(0)$, to the binary signal $i_d=\Delta_{\dmin}(i_a)$, defined by
\begin{equation}
\Delta_{\dmin}(i_a)(t) =
\begin{cases}
i_a(0-) & \text{if } t< \dmin\\
i_a(t-\dmin) & \text{if } t\geq \dmin
\enspace.
\end{cases}
\end{equation}
\item An open set $U\subseteq \IR^n$, where $\pi_1[U]$ represents the analog output signal and $\pi_k[U]$, $k=\{2,3, \ldots, n \}$, specifies the internal state variables of the model. In this fashion,\footnote{In real circuits, the
interval $(0,1)$ typically needs to be replaced by $(0,\vdd)$.} we presume that $\pi_1[U] = (0,1)$, i.e., the range of output signals is contained in the interval $(0,1)$.
\item Two bounded functions $f_\uparrow, f_\downarrow :\IR \times U \to \IR^n$ 
 with the following properties:
    \begin{itemize}
        \item $f_\uparrow, f_\downarrow$ are continuous for $(t,x) \in [0,T] \times U$, 
         for any $0 < T < \infty$, and Lipschitz continuous in $U$, which entails
that every trajectory~$x$ of the ODEs $\frac{d}{dt}\,x(t) = f_\uparrow(t,x(t))$ and $\frac{d}{dt}\,x(t) = f_\downarrow(t,x(t))$ with any initial value $x(0)\in U$ satisfies $x(t)\in U$ for all~$t\in [0,T]$, recall \cref{sec:contODE}.
        \item for no trajectory~$x$ of the ODEs $\frac{d}{dt}\,x(t) = f_\uparrow(t,x(t))$ and $\frac{d}{dt}\,x(t) = f_\downarrow(t,x(t))$ with initial value $x(0)\in U$ does $\pi_1\circ x$ have infinitely many local optima, i.e., critical points with $(\pi_1\circ x)'(t)=0$.
    \end{itemize}
\item An initial value $x_0\in U$, with $x_0=f_\uparrow$ if $i_a(0-) = 1$ and $x_0 = f_\downarrow$ if $i_a(0-)=0$.
\item A mode-switch signal $a:[0,T] \to\{f_\uparrow,f_\downarrow\}$ defined by setting $a(t) = f_\uparrow$ if $i_d(t) = 1$ and $a(t) = f_\downarrow$ if $i_d(t)=0$.
\item The analog output signal $o_a=x_{a}$, i.e., the output signal for~$a$ and initial value~$x_0$.
\item A threshold voltage $\xi=\vth \in(0,1)$ for the comparator that finally produces the binary output
signal $o_d=\Theta_\xi(o_a)$.
\end{itemize}

By combining the results from \cref{sec:contODE} and \ref{sec:thresholding}, we obtain:

\begin{thm}\label{thm:channels:are:cont}
The channel function of an IDM channel, which maps from the input signal $i_a$ to the 
output signal $o_d$, is continuous with respect to the $1$-norm on the interval~$[0,T]$.
\end{thm}
\begin{proof}
The mapping from $i_a$ to $o_d$ is continuous as the concatenation of continuous mappings:
\begin{itemize}
\item The mapping from $i_a \mapsto i_d$ is continuous since $\Delta_{\dmin}$ is trivially
continuous for input and output binary signals with the $1$-norm.
\item The mapping $i_d\mapsto a$ is a continuous mapping from the set of signals equipped with the $1$-norm to the set of mode-switch signals equipped with the metric~$d_T$, since the points of discontinuity of $a$ are the points where $i_d$ is discontinuous.
\item By Theorem~\ref{thm:real-valued:is:cont}, the mapping $a\mapsto x_{a}$ is a continuous mapping from the set of mode-switch signals equipped with the metric~$d_T$ to the set of piecewise differentiable functions $[0,T]\to U$ equipped with the supremum-norm.
\item The mapping $x_a\mapsto \pi_1\circ x_{a}$ is a continuous mapping from the set of piecewise differentiable functions $[0,T]\to U$ equipped with the supremum-norm to the set of piecewise differentiable functions $[0,T]\to (0,1)$ equipped with the supremum-norm.  Since $\lVert (x_1,\dots,x_n) \rVert_1 = 
  \lVert x_1 \rVert_1 + \dots +  \lVert x_n \rVert_1$ for every $x \in U$,
  this follows from $\lVert  \pi_1(x)\rVert_1 \leq \lVert x\rVert_1$.
\item By Theorem~\ref{thm:thr:cont}, the mapping
$\pi_1\circ x_{a} \mapsto \Theta_\xi(\pi_1\circ x_{a})$ is a continuous mapping from the set of piecewise differentiable functions $[0,T]\to (0,1)$ equipped with the supremum-norm to the set of binary signals equipped with the $1$-norm.\qedhere
\end{itemize}
\end{proof}

General digitized hybrid gates have $c\geq 1$ binary input signals $i_a=(i_a^1,\dots,i_a^c)$, augmented with
\emph{initial values} $(i_a^1(0-),\dots,i_a^c(0-))$, and a single binary output signal $o_d$, and are specified as follows:

\begin{definition}[Digitized hybrid gate]\label{def:dhm}
A digitized hybrid gate with $c$ inputs consists of:
\begin{itemize}
\item $c$ delay functions $\Delta_{\delta_j}(s)$ with $\delta_j\geq 0$, $1 \leq j \leq c$, 
that map the binary input signal
$i_a^j$ with initial value $i_a^j(0-)$ to the binary signal $i_d^j=\Delta_{\delta_j}(i_a^j)$, defined by 
\begin{equation}
\Delta_{\delta_j}(i_a^j)(t)=
\begin{cases}
i_a^j(0-) & \text{if } t< \delta_j\\
i_a^j(t-\delta_j) & \text{if } t\geq \delta_j \label{eq:idj}
\enspace.
\end{cases}
\end{equation}
\item An open set $U\subseteq \IR^n$, where $\pi_1[U]$ represents the analog output signal and $\pi_k[U]$, $k=\{2,3, \ldots, n \}$, specifies the internal state variables of the model. 
\item A set $F$ of bounded functions $f^\ell :\IR \times U \to \IR^n$,
  with the following properties:
    \begin{itemize}
    \item $f^\ell$
     is continuous for $(t,x) \in [0,T] \times U$, 
          for any $0 < T < \infty$, and Lipschitz continuous in $U$, with
          a common Lipschitz constant, which entails
          that every trajectory~$x$ of the ODE $\frac{d}{dt}\,x(t) = f^\ell(t,x(t))$
          with any initial value $x(0)\in U$ satisfies $x(t)\in U$ for all~$t\in [0,T]$.
        \item for no trajectory~$x$ of the ODEs $\frac{d}{dt}\,x(t) = f^\ell(t,x(t))$ with initial value $x(0)\in U$ does $\pi_1\circ x$ have infinitely many local optima, i.e., critical points with $(\pi_1\circ x)'(t)=0$.
     \end{itemize}
\item A mode-switch signal $a:[0,T] \to F$, which obtained by a continuous
  choice function $a_c$ acting on $i_d^1(t),\dots,i_d^c(t)$, i.e., 
  $a(t) = a_c(i_d^1(t),\dots,i_d^c(t))$. 
\item An initial value $x_0\in U$, which must correspond to the mode selected by $a_c(i_a^1(0-),\dots,i_a^c(0-))$.
\item The analog output signal $o_a=x_{a}$, i.e., the output signal for~$a$ and initial value~$x_0$.
\item A threshold voltage $\xi=\vth \in(0,1)$ for the comparator that finally produces the binary output
signal $o_d=\Theta_\xi(o_a)$.
\end{itemize}
\end{definition}

By essentially the same proof as for \cref{thm:channels:are:cont}, we obtain:
\begin{thm}\label{thm:digitizedmodels:are:cont}
The gate function of a digitized hybrid gate with $c$ inputs, which maps from the vector of input signals $i_a=(i_a^1,\dots,i_a^c)$ to the output signal $o_d$, is continuous with respect to the $1$-norm on the interval~$[0,T]$.
\end{thm}

\section{Composing Gates in Circuits}
\label{sec:circuits}

In this section, we will first compose digital circuits from digitized hybrid gates
and reason about their executions. More specifically, it will turn out that, under
certain conditions ensuring the causality of every composed gate, the resulting
circuit will exhibit a unique execution, for every given execution of its
inputs. This uniqueness is mandatory for building digital dynamic
timing simulation tools.

Moreover, we adapt the proof that no circuit with IDM
channels can solve the bounded SPF problem utilized in \cite{FNNS19:TCAD} 
to our setting: Using the continuity result of \cref{thm:digitizedmodels:are:cont}, we will
prove that no circuit with digitized hybrid gates can solve bounded SPF. Since unbounded
SPF can be solved with IDM channels, which are simple instances of digitized hybrid gate
models, faithfulness follows.

\subsection{Executions of circuits}

\textbf{Circuits.}  
Circuits are obtained by interconnecting a set of input ports and a set
   of output ports, forming the external interface of a circuit, 
   and a finite set of digitized hybrid gates. 
We constrain the way components are interconnected in a natural
way, by requiring that any gate input, channel input and output
port is attached to only one input port, gate output or channel
output, respectively. 
Formally, a {\em circuit\/} is described by a directed graph where:
\begin{enumerate}
\item[C1)] A vertex $\Gamma$ can be either a {\em circuit input port}, a {\em
    circuit output port}, or a digitized hybrid {\em gate}.
\item[C2)] The \emph{edge} $(\Gamma,I,\Gamma')$ represents a $0$-delay connection from
the output of $\Gamma$ to a fixed input $I$ of $\Gamma'$. 
\item[C3)] Circuit input ports have no incoming edges.
\item[C4)] Circuit output ports have exactly one incoming edge and no outgoing one. 
\item[C5)] A $c$-ary gate $G$ has a single output and $c$ inputs $I_1,\dots, I_c$,
  in a fixed order, fed by incoming edges from exactly one gate output or input port.
\end{enumerate}%.

\smallskip
\noindent\textbf{Executions.}  An {\em execution\/} of a circuit~$\C$ is a collection of
binary signals~$s_\Gamma$ defined on $[0,\infty)$
for all vertices~$\Gamma$ of~$\C$ that respects all
the gate functions and input port signals.  Formally, the following
properties must hold:
\begin{enumerate}
\item[E1)] If~$i$ is a circuit input port, there are no restrictions on~$s_i$.
\item[E2)] If~$o$ is a circuit output port, then~$s_o = s_G$, where~$G$ is the
unique gate output connected to~$o$.
\item[E3)] If vertex~$G$ is a gate with~$c$ inputs $I_1,\dots,I_c$, ordered
  according to the fixed order condition C5), and gate function~$f_G$, then
  $s_G = f_G(s_{\Gamma_1},\dots, s_{\Gamma_c})$, where
  $\Gamma_1,\dots,\Gamma_c$ are the vertices the inputs $I_1,\dots,I_c$ of $C$
  are connected to via edges $(\Gamma_1,I_1,G), \dotsm, (\Gamma_d,I_c,G)$.
\end{enumerate}

The above definition of an execution of a circuit
is ``existential'', in the sense that it
  only allows checking for a given collection of signals whether 
  it is an execution or not: For every hybrid gate in the
  circuit, it specifies the gate output
  signal, given a {\em fixed\/} vector of input signals, all defined on
  the time domain $t\in[0,\infty)$.
A priori, this does not give an algorithm to construct executions of circuits,
in particular, when they contain feedback loops. Indeed, the parallel
composition of general hybrid automata may lead to non-unique executions and
bizarre timing behaviors known as \emph{Zeno}, where an infinite number of 
transitions may occur in finite time~\cite{LSV03}.

To avoid such behaviors in our setting, we require all discretized
hybrid gates in a circuit to be \emph{strictly causal}:

\begin{definition}[Strict causality]\label{def:strictcausality}
A digitized hybrid gate $G$ with $c$ inputs is strictly causal, 
if the pure delays $\delta_j$ for every $1\leq j \leq c$ are
positive. Let $\dmin^C>0$ be the minimal pure delay of any input
of any gate in circuit $C$.
\end{definition}

We proceed with defining input-output causality for gates,
which is based on signal transitions. Every binary signal can equivalently be described by a sequence 
of transitions: A {\em falling transition\/} at time~$t$
is the pair $(t,0)$, a {\em rising
transition\/} at time~$t$ is the pair $(t,1)$. 

\begin{definition}[Input-output causality]\label{def:inputoutputcausality}
The output transition $(t,.)\in s_G$ of a gate G \emph{is caused}
by the transition $(t',.) \in s_G^j$ on input $I_j$ of $G$, if
$(t,.)$ occurs in the mode $a_c(i_d^1(t^+),\dots,i_d^c(t^+))$, where
$i_d^j(t^+)$ is the pure-delay shifted input signal at input $I_j$ 
at the last mode switching time $t^+\leq t$ (see \cref{eq:idj}) and
 $(t',.)$ is the last transition in $s_G^j$ before or at time
$t^+-\delta_j$, i.e., $\not\exists (t'',.) \in s_G^j$ for 
$t'<t''\leq t^+-\delta_j$. 

We also assume that the output transition $(t,.)\in s_G$
\emph{causally depends} on every transition in $s_G^j$ before 
or at time $t^+-\delta_j$.
\end{definition}

Strictly causal gates satisfy the following obvious property:

\begin{lemma}\label{lem:ioseparationtime}
If some output transition $(t,.)\in s_G$ of a strictly causal digitized hybrid gate $G$
in a circuit $C$ causally depends on its input transition $(t',.)\in s_G^j$, then $t-t' \geq \delta_j$.
\end{lemma}

The following \cref{thm:execution} shows that every circuit made up of 
strictly causal gates has a unique execution, defined for $t\in[0,\infty)$.

\begin{theorem}[Unique execution]\label{thm:execution}
Every circuit $C$ made up of finitely many strictly causal digitized hybrid
gates has a unique execution, which either consists of finitely
many transitions only or else requires $[0,\infty)$ as its time
domain.
\end{theorem}
\begin{proof}
We will inductively construct this unique execution by a sequence of iterations
$\ell \geq 1$ of a simple deterministic simulation algorithm, which 
determines the prefix of the sought execution up to time $t_\ell$. 
Iteration $\ell$ deals with transitions occurring at time $t_\ell$,
starting with $t_1=0$. To every transition $e$ generated throughout
its iterations, we also assign a \emph{causal depth} $d(e)$ that 
gives the maximum causal distance to an input port: $d(e)=0$ if 
$e$ is a transition at some input port, and $d(e)$ is the maximum
of $1 + d(e^j)$, $1 \leq j \leq c$, for every transition added at the output of 
a $c$-ary gate caused by transitions $e^j$ at its inputs.

Induction basis $\ell=1$: At the beginning of iteration 1, which deals with
all transitions occurring at time $t_1=0$,
all gates are in their initial mode, which is determined by the initial
values of their inputs. They are either connected to input ports, in which
case $s_i(0-)$ is used, or to the output port of some gate $G$, in which case
$s_G(0)$ (determined by the initial mode of $G$) is used. 
Depending on whether $s_i(0-)=s_i(0)$ or not, there is also an input transition $(0,s_i(0)) 
\in s_i$ or not. All transitions in the so generated execution prefix 
$[0,t_1]=[0,0]$ have a causal depth of 0.

Still, the transitions that have happened by time $t_1$ may cause additional \emph{potential
future transitions}. They are called future transitions, because they occur only 
after $t_1$, and potential because they need not occur in the final execution. 
In particular, if there is an input transition $(0,s_i(0)) 
\in s_i$, it may cause a mode switch of every gate $G$ that 
is connected to the input port $i$. Due to \cref{lem:ioseparationtime},
however, such a mode switch, and hence each of the output 
transitions~$e$ that may occur during that new mode (which all are assigned
a causal depth $d(e)=1$), of $G$ can only 
happen at or after time $t_1+\dmin^C$.
In addition, the initial mode of any gate $G$ that is not mode switched
may also cause output transitions $e$ at arbitrary times $t > 0$, which are assigned
a causal depth $d(e)=0$. Since at most finitely many critical points may exist
for every mode's trajectory, it follows that at most \emph{finitely} many 
such future potential transitions could be generated in each of the finitely 
many gates in the circuit.
Let $t_2>t_1$ denote the time of the closest transition among all
input port transitions and all the potential future transitions just introduced.

Induction step $\ell \to \ell+1$: Assume that the execution prefix
for $[0,t_\ell]$ has already been constructed in iterations 
$1,\dots,\ell$, with at most finitely many potential future transitions 
occurring after $t_\ell$. If the latter set is empty, then the execution
of the circuit has already been determined completely. 
Otherwise, let $t_{\ell+1}>t_\ell$ 
be the closest future transition time. 

During iteration $\ell+1$, all transitions occurring
at time $t_{\ell+1}$ are dealt with, exactly as in the base case:
Any transition $e$, with causal depth $d(e)$, 
happening at $t_{\ell+1}$ at a gate output or at some
input port may cause a mode switch of every gate $G$ that 
is connected to it. Due to \cref{lem:ioseparationtime},
such a mode switch, and hence each of the at most finitely
many output transitions $e'$ occurring during 
that new mode (which all are assigned a causal depth $d(e')=d(e)+1$), 
of $G$ can only happen at or after time $t_{\ell+1}+\dmin^C$. 
In addition, the at most finitely many potential future transitions w.r.t.
$t_{\ell}$ of all gates that have not been mode-switched and actually
occur at times greater than $t_{\ell+1}$ are retained, along with their
assigned causal depth, as potential future
transitions w.r.t. $t_{\ell+1}$. Overall, we again end up with at most 
finitely many potential future transitions, which completes the induction step.

To complete our proof, we only need to argue that 
$\lim_{\ell\to\infty} t_\ell = \infty$ for the case where the iterations
do not stop at some finite $\ell$. This follows immediately from the fact that, for every
$k\geq 1$, there must be some iteration $\ell \geq 1$ such that 
$t_{\ell} \geq k\dmin^C$. If this was not the case, there must be
some iteration after which no further mode switch of any gate takes place.
This would cause the set of potential future transitions to shrink in
every subsequent iteration, however, and thus the simulation algorithm to stop,
which provides the required contradiction.
\end{proof}

From the execution construction, we also immediately get:

\begin{lemma}\label{lem:Depth_Iteration}%.
For all~$\ell\ge 1$, (a) the simulation algorithm never assigns a
     causal depth larger than~$\ell$ to a transition generated in
     iteration~$\ell$, and (b) at the end of iteration~$\ell$ the sequence
     of causal depths of transitions in~$s_{\Gamma}$ for $t\in[0,t_\ell]$ is
     nondecreasing for all components~$\Gamma$.
\end{lemma}

\subsection{Impossibility of short-pulse filtration}

The results of the previous subsection allow us to adapt the impossibility proof
of \cite{FNNS19:TCAD} to our setting. We start with the the definition of the
SPF problem:

\textbf{Short-Pulse Filtration.} 
A signal {\em contains a pulse\/} of length~$\Delta$ at time~$T_0$, if it
 contains a rising transition at time~$T_0$, a falling transition at time
 $T_0+\Delta$, and no transition in between.
The \emph{zero signal} has the initial value 0 and does not contain any transition.
A circuit {\em solves Short-Pulse Filtration (SPF)\/}, if it fulfills all of:
\begin{enumerate}
\item[F1)] The circuit has exactly one input port and exactly one output port. {\em
        (Well-formedness)}
\item[F2)] If the input signal is the zero signal, then so is the output
	signal. {\em (No generation)}
\item[F3)] There exists an input pulse such that
	the output signal is not the zero signal. {\em (Nontriviality)}
\item[F4)] There exists an~$\varepsilon>0$ such that for every input pulse the output
	signal never contains a pulse of length less than or equal to~$\varepsilon$. {\em
	(No short pulses)}
\end{enumerate}
We allow the circuit to behave arbitrarily if the input
  signal is not a single pulse or the zero signal.

A circuit {\em solves bounded SPF\/} if additionally:
\begin{enumerate}
\item[F5)] There exists a $K>0$ such that for every input pulse
  the last output transition
  is before time~$T_0+\Delta+K$, where~$T_0$ is the time of the first input transition.
  {\em (Bounded stabilization time)}
\end{enumerate}

A circuit is called a {\em forward circuit\/} if its graph is acyclic.
Forward circuits are exactly those circuits that do not contain
     feedback loops.
Equipped with the continuity of digitized hybrid gates and the fact that
     the composition of continuous functions is continuous, it is not
     too difficult to prove that the inherently discontinuous SPF problem
     cannot be solved with forward circuits.

\begin{theorem}\label{thm:no_forward_circuit}
No forward circuit solves bounded SPF.
\end{theorem}
\begin{proof}
Suppose that there exists a forward circuit that solves bounded SPF with
stabilization time bound~$K$.
Denote by~$s_\Delta$ its output signal when feeding it a $\Delta$-pulse at
time~$0$ as the input.
Because~$s_\Delta$ in forward circuits is a finite composition of continuous
functions by
Theorem~\ref{thm:digitizedmodels:are:cont}, $\lVert s_\Delta \rVert_{[0,T],1}$
depends continuously on~$\Delta$, for any $T$.

By the nontriviality condition (F3) of the SPF problem, there exists
some~$\Delta_0$ such that $s_{\Delta_0}$ is not the zero signal.
Set $T = 2\Delta_0 + K$.

Let~$\varepsilon>0$ be smaller than both~$\Delta_0$ and $\lVert s_{\Delta_0} \rVert_{[0,T],1}$.
We show a contradiction by finding some~$\Delta$ such that~$s_\Delta$ either
contains a pulse of length less than~$\varepsilon$ (contradiction to the no
short pulses condition (F4)) or contains a transition
after time $\Delta+K$ (contradicting the bounded stabilization time
condition~(F5)).

Since $\lVert s_\Delta \rVert_{[0,T],1}\to0$ as $\Delta\to0$ by the no generation condition (F2)
of SPF,
there exists a~$\Delta_1<\Delta_0$ such that $\lVert s_{\Delta_1} \rVert_{[0,T],1}=\varepsilon$
by the intermediate value property of continuity.
By the bounded stabilization time condition (F5), there are no transitions
in~$s_{\Delta_1}$ after time $\Delta_1+K$.
Hence, $s_{\Delta_1}$ is~$0$ after this time because otherwise it is~$1$ for the
remaining duration $T - (\Delta_1+K) > \Delta_0 > \varepsilon$, which would
mean that $\lVert s_{\Delta_1} \rVert_{[0,T],1}>\varepsilon$.
Consequently, there exists a pulse in~$s_{\Delta_1}$ before
time $\Delta_1+K$.
But any such pulse is of length at most~$\varepsilon$ because
$\lVert s_{\Delta_1} \rVert_{[0,\Delta_1+K],1}  \leq 
\lVert s_{\Delta_1} \rVert_{[0,T],1}=\varepsilon$.
This is a contradiction to the no short pulses condition (F4).
\end{proof}

We next show how to simulate (part of) an execution of an arbitrary
     circuit~$\C$ by a forward circuit~$\C'$ generated from~$\C$ by
     the unrolling of feedback loops.
Intuitively, the deeper the unrolling, the longer the time~$\C'$
     behaves as~$\C$.

\begin{definition}%.
Let~$\C$ be a circuit, $V$ a vertex of~$\C$, and $k\geq0$.
We define the {\em $k$-unrolling of~$\C$ from~$V$}, denoted by~$\C_k(V)$, 
to be a directed acyclic graph with a single sink, constructed
as follows:

The unrolling $\C_k(I)$ from input port $I$ is just a copy of that input port.
The unrolling $\C_k(O)$ from output port $O$ with incoming channel $C$ and predecessor $V$ comprises a copy of the output port $O^{(k)}$ and the unrolled circuit $\C_k(V)$ with its sink connected to $O^{(k)}$ by an edge. 

The 0-unrolling $\C_0(B)$ from hybrid gate $B$ is a trivial Boolean gate $X_v$ without inputs and the constant output value~$v$ equal to $B$'s initial digitized output value. For $k>0$, the $k$-unrolling $\C_k(B)$ from gate $B$ comprises an exact copy of that gate $B^{(k)}$. 
Additionally, for every incoming edge of $B$ from $V$ in $\C$, it contains the circuit $\C_{k-1}(V)$  with its sink connected to $B^{(k)}$.
Note that all copies of the same input port are considered to be the same.
\end{definition}

To each component~$\Gamma$ in~$\C_k(V)$, we assign a value~$z(\Gamma)
\in\IN_0\cup\{\infty\}$ as follows:
     $z(\Gamma) = \infty$
     if~$\Gamma$ has no predecessor (in particular, is an input port) 
and $\Gamma\not\in\{X_0,X_1\}$. Moreover, $z(X_0)=z(X_1) = 0$, 
     $z(V)=z(U)$ if $V$ is an output port connected by an edge to $U$, and
     $z(B) = \min_{c\in E^B}\{1+z(c)\}$ if $B$ is a gate with its inputs connected
to the components in the set $E^B$.
\cref{fig:unrolling} shows an example of a circuit and an
     unrolled circuit with its $z$~values.

\begin{figure}
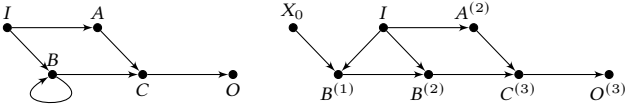

\centering
  {
    \tikzfigureunrolling}
\caption{Circuit~$\C$ (left) and $\C_3(O)$ (right) under the assumption that the gate~$B$ has initial value~$0$. It is $z(X_0)=0$,
	$z(I)=z(A^{(2)})=\infty$, $z(B^{(1)})=1$, $z(B^{(2)})=2$, $z(C^{(3)})=3$, and
$z(O^{(3)})=3$.}
\label{fig:unrolling}%\vspace{-0.4cm}
\end{figure}

Noting that, for every component $\Gamma$ in $C_k(V)$, $z(\Gamma)$ is the number of gates on the shortest path from an $X_v$ node to $\Gamma$, or $z(\Gamma)=\infty$ if no such path exists, we immediately get:

\begin{lemma}\label{lem:unrol:kisz}
  The $z$-value assigned to the sink vertex $V^{(k)}$ of a $k$-unrolling $\C_k(V)$ of $\C$ from $V$ satisfies $z(V^{(k)})\geq k$.
\end{lemma}

Recalling the causal depths assigned to transitions during the execution construction in \cref{thm:execution}, we are now in the position to prove the result for a circuit simulated by an unrolled circuit.

\begin{theorem}\label{thm:simulation}%.
Let~$\C$ be a circuit with input port~$I$ and output port~$O$ that solves
bounded SPF.
Let~$\C_k(O)$ be an unrolling of~$\C$, $\Gamma$ a component in~$\C$, and
     $\Gamma'$ a copy of $\Gamma$ in~$\C_k(O)$.
For all input signals~$s_I$ on $I$, if a transition~$e$ is generated for~$\Gamma$
 by the execution construction algorithm run on
     circuit~$\C$ with input signal~$s_I$
     and~$d(e) \le z(\Gamma')$, then~$e$ is also generated
     for~${\Gamma'}$ by the algorithm run on circuit~$\C_k(O)$
     with input signal~$s_I$; and vice versa.
\end{theorem}%.
\begin{proof}
Assume that $e$ is the first transition violating the theorem.
The input signal is the same for both circuits, and the initial digitized values of gates in $\C$ and both their copies in $\C_k(O)$ and the $X_v$ gates resulting from their $0$-unrolling are equal as well. Hence, $e$ cannot be any such transition (added in iteration 1 only).

If $e$ was added to the output of a gate $B$ in either circuit, the transition $e'$ resp.\ $e''$ at one of its inputs that caused $e$ in $\C$ resp.\ $\C_k(V)$ must have been different. 
These transitions $e'$ resp.\ $e''$ must come from the output of some other gate $B_1$,
and causally precede $e$. Hence,
by \cref{def:inputoutputcausality}, $d(e)=d(e')+1$, and by \cref{lem:Depth_Iteration}, $d(e)\geq d(e'')$. Also by definition, $z(B)=z(B_1)+1$ in $C_k(O)$. Since $d(e)\leq z(B)$ by assumption, 
we find $d(e')\leq z(B_1)$ and $d(e'') \leq z(B)$, so applying our theorem to 
$e'$ and $e''$ yields a contradiction to $e$ being the first violating transition.
\end{proof}

We can finally prove that bounded SPF is not solvable, even with
non-forward circuits.

\begin{theorem}\label{thm:main_impossibility}
No circuit solves bounded SPF.
\end{theorem}

\begin{proof}
We first note that the impossibility of bounded SPF also implies the
impossibility of bounded SPF when restricting pulse lengths to be at most some
$\Delta_0>0$.

Since all transitions generated in the execution construction
\cref{thm:execution} up to any bounded time $t_\ell$ have bounded causal depth,
let~$\zeta$ be an upper bound on the causal depth of transitions up to the
SPF stabilization time bound~$\Delta_0+K$.
Then, by Theorem~\ref{thm:simulation} and Lemma~\ref{lem:unrol:kisz}, the $\zeta$-unrolled
circuit~$\C_\zeta(O)$ has the same output transitions
as the original circuit~$\C$ up to time~$\Delta_0+K$, and hence, by definition of
bounded SPF, the same transitions for all times.
But since~$\C_\zeta(O)$ is a forward circuit, it cannot solve bounded SPF by
Theorem~\ref{thm:no_forward_circuit}, i.e., neither can~$\C$.        
\end{proof}

\section{Applications}
\label{sec:examples}

We next discuss three examples of thresholded mode-switched ODE systems.
For all non-closed systems, the proven continuity shows that similar digital inputs
  lead to similar digital outputs.

We start with an introductory example from control theory,
  the bang-bang heating controller for thermodynamic systems.
Following \cite{henzinger2000theory}, let $x(t)$ be the system's temperature at time~$t$
  and $h(t)$ be the mode of the binary heating signal that can be off (0)
  or on (1).
With a pure delay $\delta > 0$ for the heating to take effect,
  we assume that the heat flow is described as
\begin{align}
\dot{x} = 
\begin{cases}
 -0.1 x(t) & \text{if } h(t-\delta) = 0\\
5-0.1 x(t) & \text{if } h(t-\delta) = 1\\
\end{cases}
\end{align}
for the heating being off or on, respectively: 
the temperature falls to $0$ in the former case and approaches
$50$ degrees in the latter.

The heating signal is controlled by a bang-bang controller (with hysteresis) 
with two threshold temperatures, 19 and 20 degrees. It could be implemented by an ideal
SR-latch, with pure delay $\delta$, where the Set port (S) is driven by the inverted
$\neg \Theta_{19}(x)$,
the reset port (R) is driven by $\Theta_{21}(x)$, and
the output of the latch controlling the heating mode signal~$h$.

\medskip

In fact, digital circuits are a particularly rich and interesting source of 
application examples in general. We will demonstrate this by means of 
two hybrid gate models for a CMOS \NOR\ gate (see \cref{fig:nor_CMOS-twoCap} for 
the schematics), namely, the simple model proposed in \cite{FMOS22:DATE} (as an instance
  of an autonomous ODE model) and the advanced model presented in \cite{ferdowsi2022accurate}
 (as an instance of a non-autonomous ODE model). The SR-latch from the previous 
example can be implemented via two cross-coupled \NOR\ gates.

\begin{figure}[t!]
  \centering
    \includegraphics[height=0.45\linewidth]{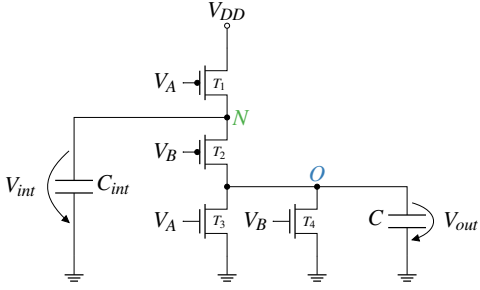}%
  \caption{Transistor level implementation of the \NOR\ gate.}
     \label{fig:nor_CMOS-twoCap}
\end{figure}

\subsection{Simple Hybrid Model}
\label{sec:SimpleModel}

The \emph{simple hybrid gate model} proposed in \cite{FMOS22:DATE} replaces all transistors in \cref{fig:nor_CMOS-twoCap} by ideal zero-time switches, which are switched on and off at the relevant input threshold voltage $V_{th} = \vdd /2$ crossing times. More precisely, depending on whether the corresponding input is 1 or 0, 
every pMOS transistor is removed ($R = \infty$) resp.\ replaced by a fixed resistor $R < \infty$, and vice versa for an nMOS transistor. This leads to the following system of coupled autonomous first-order ODEs governing the analog trajectories of the gate’s output in the respective mode:

\begin{itemize}
\item System $(1,1)$: $\va=1$, $\vb=1$:
If inputs $A$ and $B$ are 1, both \nmos\ transistors are conducting and thus replaced by resistors, causing the output $O$ to be discharged in parallel. By contrast, $N$ is
completely isolated and keeps its value. This leads to the following ODEs:
\begin{align*}
  &\ \displaystyle{{\operatorname{d}\over\operatorname{d}\!t} \vint{}(t) \choose {\operatorname{d}\over\operatorname{d}\!t} \vout(t)}=
  \displaystyle{f_1(\vint{}(t), \vout(t)) \choose f_2(\vint{}(t), \vout(t))}=
  \displaystyle{0  \choose - \bigl(\frac{1}{\cout R_3} + \frac{1}{\cout R_4}\bigr) \vout(t)}
\end{align*}

\item System $(1,0)$: $\va=1$, $\vb=0$:
Since $T_1$ and $T_4$ are open, node $N$ is connected to $O$, and $O$ to \gnd. Both capacitors have to be discharged over resistor $R_3$, resulting in less current that is available for discharging~$\cout$. One obtains:
\begin{align*}
  &\ \displaystyle{{\operatorname{d}\over\operatorname{d}\!t} \vint{}(t) \choose {\operatorname{d}\over\operatorname{d}\!t} \vout(t)}=
  \displaystyle{f_3(\vint{}(t), \vout(t)) \choose f_4(\vint{}(t), \vout(t))}=
  \displaystyle{-\frac{\vint{}(t)}{\cint{} R_{2}}  +  \frac{\vout(t)}{\cint{} R_{2}}  \choose \frac{\vint{}(t)}{\cout R_{2}}  - \bigl(\frac{1}{\cout R_2} + \frac{1}{\cout 
                               R_3}\bigr) \vout(t) }
\end{align*}

\item System $(0,1)$: $\va=0$, $\vb=1$:
Opening transistors $T_2$ and $T_3$ again decouples the nodes $N$ and $O$. We thus get
\begin{align*}
  &\ \displaystyle{{\operatorname{d}\over\operatorname{d}\!t} \vint{}(t) \choose {\operatorname{d}\over\operatorname{d}\!t} \vout(t)}=
  \displaystyle{f_5(\vint{}(t), \vout(t)) \choose f_6(\vint{}(t), \vout(t))}=
  \displaystyle{-\frac{\vint{}(t)}{\cint{} R_{1}} +  \frac{\vdd}{\cint{} R_1}  \choose - \frac{\vout(t)}{\cout R_4} }
\end{align*}

\item System $(0,0)$: $\va=0$, $\vb=0$:
Closing both \pmos\ transistors causes both capacitors to be charged over the same resistor~$R_1$, similarly to system $(1,0)$. Thus
\begin{align*}
  &\ \displaystyle{{\operatorname{d}\over\operatorname{d}\!t} \vint{}(t) \choose {\operatorname{d}\over\operatorname{d}\!t} \vout(t)}=
  \displaystyle{f_7(\vint{}(t), \vout(t)) \choose f_8(\vint{}(t), \vout(t))}=\\
  &\ \displaystyle{-\bigl(\frac{1}{\cint{}(t) R_1} + \frac{1}{\cint{}(t) R_2}\bigr) \vint{} + \frac{\vout(t)}{\cint{} R_2} +  \frac{\vdd}{\cint{} R_1}  \choose \frac{\vint{}(t)}{\cout R_2}  - \frac{\vout(t)}{\cout R_2} }
\end{align*}
\end{itemize}

Every $f_i$, $i \in \{1, \ldots, 8 \}$, is a mapping from $U=(0,1)^2 \subseteq \IR^2$ to $\IR$, whereat $U$ is the vector of the voltages at the nodes $N$ and $O$ in \cref{fig:nor_CMOS-twoCap}. Solving the above ODEs provides analytic expressions for these voltage trajectories, which can even be inverted to obtain the relevant gate delays.
As it turned out in \cite{FMOS22:DATE}, although the model perfectly covers the MIS effects in the case of falling output transitions, it fails to do so in the rising output transitions case. Nevertheless, despite this accuracy shortcoming, the results of the present paper imply that the model is faithful. More specifically, 
we obtain the following theorem:

\begin{theoremrep}[]\label{thm:SimpleModel}
For any $i \in \{1, \ldots, 8 \}$, the mapping $f_i$, defined above, is Lipschitz continuous.
\end{theoremrep}
%\begin{thm}
%For any $i \in \{1, \ldots, 8 \}$, the mapping $f_i$, defined above, is Lipschitz continuous.
%\end{thm}
\begin{proof}
Albeit the proof is evident, we elaborate it for $f_7$; similar arguments apply to the other cases. Let $K= max \bigl\{ (\frac{1}{\cint{} R_1} + \frac{1}{\cint{} R_2}) , \frac{1}{\cint{} R_2} \bigr\}$.
For any voltages of $\vint{}(t)^{(1)}$, $\vint{}(t)^{(2)}$, $\vout(t)^{(1)}$, and $\vout(t)^{(2)}$ belonging to $(0,1)$, we find
\begin{align*}
&\bigl\lVert f_7(\vint{}(t)^{(1)}, \vout(t)^{(1)}) - f_7(\vint{}(t)^{(2)}, \vout(t)^{(2)}) \bigr\rVert =\\
& \bigl\lVert -\bigl(\frac{1}{\cint{} R_1} + \frac{1}{\cint{} R_2}\bigr)(\vint{}(t)^{(1)}- \vint{}(t)^{(2)}) +\\
& \quad\frac{1}{\cint{} R_2}\bigl(\vout(t)^{(1)}- \vout(t)^{(2)}) \bigr\rVert \leq  \\
& K \bigl\lVert (\vint{}(t)^{(1)}- \vint{}(t)^{(2)}) + (\vout(t)^{(1)}- \vout(t)^{(2)}) \bigr\rVert.
\end{align*}
\end{proof}

Consequently, we can instantiate \cref{def:dhm} with
\begin{align}
a_c(i_d^A,i_d^B) = \begin {cases}
\displaystyle{f_1(\vint{}(t), \vout(t)) \choose f_2(\vint{}(t), \vout(t))} &   \ \ (i_d^A,i_d^B)=(1,1)\\
\displaystyle{f_3(\vint{}(t), \vout(t)) \choose f_4(\vint{}(t), \vout(t))} &   \ \ (i_d^A,i_d^B)=(1,0)\\
\displaystyle{f_5(\vint{}(t), \vout(t)) \choose f_6(\vint{}(t), \vout(t))} &   \ \ (i_d^A,i_d^B)=(0,1)\\
\displaystyle{f_7(\vint{}(t), \vout(t)) \choose f_8(\vint{}(t), \vout(t))} &   \ \ (i_d^A,i_d^B)=(0,0) \nonumber
\end {cases}
\end{align}

\subsection{Advanced Hybrid Model}
\label{sec:AdvancedModel}

Unlike the simple hybrid model \cite{FMOS22:DATE} outlined in the previous section, 
the \emph{advanced hybrid gate model} proposed in \cite{ferdowsi2022accurate} 
covers all MIS delay behaviors properly. It can be viewed as a generalization of the simple model, in which switching-on the pMOS transistors is not instantaneous but instead governed by a simple time evolution function representing the Shichman-Hodges transistor model~\cite{ShichmanHodges}. To be more specific, the idea is to replace the transistors with time-variant resistors (see \cref{FigureNOR-GATE}), so that the values of $R_i(t)$, $i \in \{1,\ldots,4 \}$, vary between some fixed on-resistance $R_i$ and the off-resistance $\infty$, according to the following functions:
\begin{align}
R_i^{\on}(t) &= \frac{\alpha_i}{t-t^{\on}}+R_i; \ t \geq t^{\on}, \label{on_mode}\\
R_i^{\off}(t) &= \infty; \ t \geq t^{\off}. \label{off_mode}
\end{align}
Herein, $\alpha_i$ [\si{\ohm\s}] and on-resistance $R_i$ [\si{\ohm}] are constant slope parameters; $t^{\on}$ resp.\ $t^{\off}$ represent the time when the respective transistor is switched on
resp.\ off. The switching-on of the nMOS transistors happens instantaneously also here, so $\alpha_3=\alpha_4=0$.

\begin{figure}[t!]
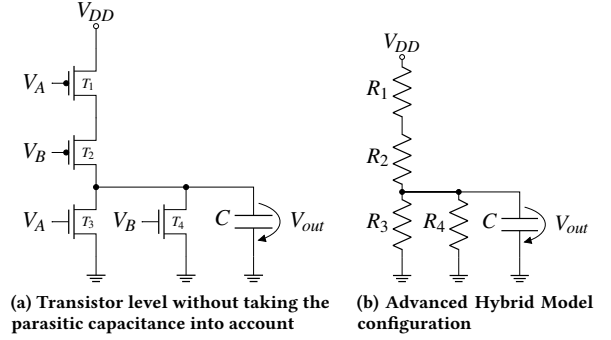

  \centering
  \subfloat[Transistor level without taking the parasitic capacitance into account]{
    \includegraphics[height=0.45\linewidth]{\figPath{nor_RC.pdf}}%
    \label{fig:nor_CMOS}}
  \hfil
  \subfloat[Advanced Hybrid Model configuration]{
    \includegraphics[height=0.4\linewidth]{\figPath{nor_R.pdf}}%
    \label{FigureNOR-GATE}}
  \caption{Implementations of a CMOS \NOR\ gate.}
\end{figure}

Applying Kirchhoff's rules to \cref{FigureNOR-GATE} leads to
$ C\frac{\dd \vout}{\dd t}= \frac{\vdd-\vout}{R_1(t)+R_2(t)} -
\frac{\vout}{R_3(t)\ ||\ R_4(t)}\ ,$
which can be transformed into the non-homogeneous non-autonomous ODE with non-constant coefficients
\begin{equation}
\label{Eq0}
\frac{\dd V_{out}}{\dd t} = f(t, V_{out}(t)) =-\frac{V_{out}(t)}{C\,R_g(t)}+U(t),
\end{equation}
where
$\frac{1}{R_g(t)}=\frac{1}{R_1(t)+R_2(t)}+\frac{1}{R_3(t)}+\frac{1}{R_4(t)}$ and
$U(t)=\frac{V_{DD}}{C(R_1(t)+R_2(t))}$. 

As comprehensively described in \cite{ferdowsi2022accurate}, depending on each particular resistor's mode in each input state transition, different expressions for $R_g(t)$, $U(t)$ and thus for $f(t,\vout(t))$ are obtained. They 
are summarized in \cref{tab:T1}. Note that we have used the notation $R_1=R_{p_A}$, $R_2=R_{p_B}$ with abbreviations $2R= R_{p_A}+R_{p_B}$, $R_3=R_{n_A}$, and $R_4=R_{n_B}$ for the two nMOS transistors $T_3$ and $T_4$. Due to the symmetry, we end up with only six different functions. 

The following theorem shows that they are continuous in the first argument and Lipschitz continuous in the second argument. 

\begin{table}[t]
\caption{$f(t,V_{out}(t))$ for each state transition.} 
\scalebox{0.94}{
\begin{tabular}{ccc}
\hline
State transition          &  & $f(t,V_{out}(t))$                                                                                                                                              \\ \cline{1-1} \cline{3-3} 
$(0,0) \rightarrow (1,0)$ &  & $f_1(t,V_{out}(t)) \doteq \frac{-V_{out}(t)}{CR_{nA}}$                                                                                                         \\
$(1,1) \rightarrow (1,0)$ &  & $f_1(t,V_{out}(t)) \doteq \frac{-V_{out}(t)}{CR_{nA}}$                                                                                                         \\
$(0,1) \rightarrow (1,0)$ &  & $f_1(t,V_{out}(t)) \doteq \frac{-V_{out}(t)}{CR_{nA}}$                                                                                                         \\
$(0,0) \rightarrow (0,1)$ &  & $f_2(t,V_{out}(t)) \doteq \frac{-V_{out}(t)}{CR_{nB}}$                                                                                                         \\
$(1,1) \rightarrow (0,1)$ &  & $f_2(t,V_{out}(t)) \doteq \frac{-V_{out}(t)}{CR_{nB}}$                                                                                                         \\
$(1,0) \rightarrow (0,1)$ &  & $f_2(t,V_{out}(t)) \doteq \frac{-V_{out}(t)}{CR_{nB}}$                                                                                                         \\
$(1,0) \rightarrow (0,0)$ &  & $f_3(t,V_{out}(t)) \doteq \frac{ \bigl (-V_{out}(t) +\vdd \bigr )}{C(\frac{\alpha_1}{t}+\frac{\alpha_2}{t+ \Delta}+ 2R)}$ \\
$(0,1) \rightarrow (0,0)$ &  & $f_4(t,V_{out}(t)) \doteq   \frac{\bigl (-V_{out}(t) +\vdd \bigr )}{C(\frac{\alpha_1}{t+ \Delta}+\frac{\alpha_2}{t}+ 2R)}$\\
$(1,1) \rightarrow (0,0)$ &  & $f_5 (t, V_{out}(t)) \doteq  \frac{\bigl (-V_{out}(t) +\vdd \bigr )t}{C ( 2Rt+ \alpha_1+\alpha_2 )}$                                                           \\
$(1,0) \rightarrow (1,1)$ &  & $f_6(t,V_{out}(t)) \doteq \frac{-V_{out}(t)}{C}(\frac{1}{R_{n_A}}+\frac{1}{R_{n_B}})$                                                                          \\
$(0,1) \rightarrow (1,1)$ &  & $f_6(t,V_{out}(t)) \doteq \frac{-V_{out}(t)}{C}(\frac{1}{R_{n_A}}+\frac{1}{R_{n_B}})$                                                                          \\
$(0,0) \rightarrow (1,1)$ &  & $f_6(t,V_{out}(t)) \doteq \frac{-V_{out}(t)}{C}(\frac{1}{R_{n_A}}+\frac{1}{R_{n_B}})$                                                                          \\ \hline
\end{tabular}}
\label{tab:T1}
\end{table}

\begin{theoremrep}[]\label{thm:AdvancedModel}
Let $F= \{f_1 , \ldots, f_6 : \IR \times (0,1) \rightarrow \IR \}$ be the set of all functions described in \cref{tab:T1}, modulo symmetry. Every $f_i \in F$, where $i \in \{1, \ldots, 6 \}$, is continuous and Lipschitz continuous in the second argument $V_{out}(t)$.
\end{theoremrep}
%\begin{thm}
%Let $F= \{f_1 , \ldots, f_6 : \IR \times (0,1) \rightarrow \IR \}$ be the set of all functions described in \cref{tab:T1}, modulo symmetry. Every $f_i \in F$, where $i \in \{1, \ldots, 6 \}$, is continuous and Lipschitz continuous in the second argument $V_{out}(t)$.
%\end{thm}
\begin{proof}
The proof is clear for functions $f_1$, $f_2$, and $f_6$. It is also straightforward for $f_5$: Let $g(t) \doteq \frac{t}{C(2Rt+ \alpha_1+ \alpha_2)}$. Since $t \in [0,T]$, $g(t)$ takes its supremum value in the interval, which we denote by $K$ (i.e., $sup_{t \in [0,T]} g(t)=K$). We observe
\begin{align*}
&\lVert f_5(t, V_{out}^{1}(t)) - f_5(t, V_{out}^{2}(t)) \rVert = \\
& \Bigl\lVert \frac{\bigl (-V_{out}^{1}(t) +\vdd \bigr ) t}{C(2Rt+ \alpha_1+ \alpha_2)}- \frac{\bigl (-V_{out}^{2}(t) +\vdd \bigr ) t}{C(2Rt+ \alpha_1+ \alpha_2)} \Bigr\rVert=\\
& \Bigl\lVert \frac{-t}{C(2Rt+ \alpha_1+ \alpha_2)} \cdot (V_{out}^{1}(t)- V_{out}^{2}(t)) \Bigr\rVert \leq | K | \bigl\lVert (V_{out}^{1}(t)- V_{out}^{2}(t)) \bigr\rVert,
\end{align*}
which concludes the proof for $f_5$. The proof for $f_3$ and $f_4$ follows the same route; we only sketch the proof of 
the Lipschitz continuity for $f_3$. We can write
\begin{align*}
& \lVert f_3(t, V_{out}^{1}(t)) - f_3(t, V_{out}^{2}(t)) \rVert =
& \bigl\lVert \frac{-(V_{out}^{1}(t)- V_{out}^{2}(t))}{\frac{\alpha_1}{t+ \Delta}+\frac{\alpha_2}{t}+ 2R} \bigr\rVert.
\end{align*}
The fact that $t$ and $\Delta$ both belong to the closed interval $[0,T]$ provides us with a Lipschitz constant $L$, 
which is independent of $t$. Consequently,
\begin{align*}
& \lVert f_3(t, V_{out}^{1}(t)) - f_3(t, V_{out}^{2}(t)) \rVert \leq  L \cdot \lVert  \bigl ( V_{out}^{1}(t)- V_{out}^{2}(t) \bigr) \rVert,
\end{align*}
which completes the proof.
\end{proof}

Defining $s(t)=(i_d^A(t^+),i_d^B(t^+))$ and $s_p(t)=(i_d^A(t),i_d^B(t))$, we can again instantiate \cref{def:dhm} by the choice function

\begin{align}
a_c(s(t)) = \begin {cases}
f_1(t, V_{out}(t)) &   \ \ s(t)=(1,0)\\ 
f_2(t, V_{out}(t)) &  \ \  s(t)=(0,1) \\
f_3(t, V_{out}(t)) &  \ \ s(t)=(0,0), s_p(t)=(1,0)\\
f_4(t, V_{out}(t)) &  \ \ s(t)=(0,0), s_p(t)=(0,1)\\
f_5(t, V_{out}(t)) &  \ \ s(t)=(0,0), s_p(t)=(1,1)\\
f_6(t, V_{out}(t)) &  \ \ s(t)=(1,1) \nonumber
\end {cases}
\end{align}
which, according to \cref{Eq0}, results in $d V_{out}(t)/dt$ being

\begin{align}
\begin {cases}
\frac{-V_{out}(t)}{CR_{nA}} &   \ \ s(t)=(1,0)\\ 
\frac{-V_{out}(t)}{CR_{nB}} &  \ \  s(t)=(0,1) \\
\frac{\bigl (-V_{out}(t) +\vdd \bigr ) t(t+\Delta_t)}{C \bigl( 2Rt^2+(\alpha_1+\alpha_2+2 \Delta_t R)t + \alpha_1 \Delta_t \bigr)} &  \ \ s(t)=(0,0), s_p(t)=(1,0)\\
\frac{\bigl (-V_{out}(t) +\vdd \bigr ) t(t+\Delta_t)}{C \bigl( 2Rt^2+(\alpha_1+\alpha_2+2 \Delta_t R)t + \alpha_2 \Delta_t \bigr)} &  \ \ s(t)=(0,0), s_p(t)=(0,1)\\
\frac{\bigl (-V_{out}(t) +\vdd \bigr ) t}{C ( 2Rt+ \alpha_1+\alpha_2  )}  &  \ \ s(t)=(0,0), s_p(t)=(1,1)\\
\frac{-V_{out}(t)}{C}(\frac{1}{R_{n_A}}+\frac{1}{R_{n_B}}) &  \ \ s(t)=(1,1). \nonumber
\end {cases}
\end{align}

\section{Conclusions}	
\label{sec:conclusions}

We presented a general continuity proof for a broad class of first-order thresholded
hybrid models, as they arise naturally in digital circuits. We showed that, under
mild conditions regarding causality, digitized hybrid gates could be composed to form
circuits with unique and well-behaved executions. We concluded with concrete gate model instantiations of our model.

\vspace{-0.2cm}

\bibliographystyle{ACM-Reference-Format}
\bibliography{mybib}

\end{document}

%% file: common.tex
\newcommand{\dup}{\delta_\uparrow}
\newcommand{\ddo}{\delta_\downarrow}
\newcommand{\fup}{f_\uparrow}
\newcommand{\fdo}{f_\downarrow}

\newcommand{\figPath}[1]{figures/#1}

\newcommand{\NOR}{\texttt{NOR}}

\newcommand{\dmin}{\delta_{\mathrm{min}}}

\newcommand{\vth}{V_{th}}

\newcommand{\vdd}{V_{DD}}
\newcommand{\gnd}{\textit{GND}}
\newcommand{\vout}{V_{out}}

\newcommand{\dd}{\mathrm{d}}

\newcommand{\va}{V_{A}}
\newcommand{\vb}{V_{B}}

\newcommand{\vint}{V_{int}}
\newcommand{\cint}{C_{int}}

\newcommand{\cout}{C}

\newcommand{\nmos}{nMOS}
\newcommand{\pmos}{pMOS}

\newcommand{\ohm}{(OHM)}

\newcommand{\on}{\mbox{\emph{on}}}
\newcommand{\off}{\mbox{\emph{off}}}

\def\C{{\mathcal C}}
\newcommand{\IN}{\mathds{N}}

%% file: defstikz.tex
\usetikzlibrary{petri}
\tikzstyle{binary place}=[place,circle, double]
\tikzstyle{node}=[circle,draw=black,thick,minimum size=9mm]
\tikzstyle{dest}=[circle,draw=black!50,fill=black!20,thick,minimum
  size=9mm]
\tikzstyle{post}=[->,thick]
\tikzstyle{pre}=[<-,thick]
\tikzstyle{every transition}=[fill,minimum width=1cm,minimum height=2mm]
\tikzstyle{Atransition}=[transition,fill,minimum width=1cm,minimum height=2mm]
\tikzstyle{Otransition}=[transition,fill=white,minimum width=1cm,minimum height=2mm]
\tikzstyle{THtransition}=[transition,fill=white,minimum width=4mm,minimum height=1cm]
\tikzstyle{Tdelay} = [draw, rectangle, rounded corners,
    minimum height=4mm, minimum width=20mm]

\tikzstyle{Tfunction} = [draw, rectangle,
    minimum height=4mm, minimum width=4mm]

\tikzstyle{Tsignal} = [draw,fill=black,circle, size=1mm]

\tikzstyle{ra} = [draw,thick,double,double distance=1.0pt,->]
\tikzstyle{r} = [draw,->,line width=0.5pt]

%% file: tikzfigures.tex
\def\figanfup[#1]{(1-exp(-(#1)/0.15))*(1-exp(-(#1)/0.16))}
\def\figanfdown[#1]{(1-(1-exp(-((#1)^2)/0.01))*(1-exp(-((#1)^2)/0.1)))}

\def\figurangle{300}
\def\figurlen{1.2}
\def\figurdist{4.5}

\def\tikzfigureunrolling{
\begin{tikzpicture}[
gate/.style={draw, fill, circle, minimum size=3pt, inner sep=0},
xscale=1,
yscale=0.6,
>=latex'
]

\path
(0,0) node (I) [gate] {}
+(0:\figurlen) node (A) [gate] {}
++(\figurangle:\figurlen) node (B) [gate] {}
++(0:\figurlen) node (C) [gate] {}
++(0:\figurlen) node (O) [gate] {};

\foreach \from/\to in {I/A, A/C, I/B, B/C, C/O} {
\path [draw, ->] (\from) -- (\to);
}
\path [draw, <-, scale=0.8]
(B) .. controls +(-1,-1) and +(1,-1) .. (B);

\foreach \what/\anchor in {I/south, A/south, B/south, C/north, O/north} {
\node [anchor=\anchor] at (\what) {\footnotesize $\what$};
}

\pgftransformshift{\pgfpoint{+0.5cm}{0cm}} 

\path
(\figurdist,0) node (i1) [gate] {}
++(0:\figurlen) node (x1) [gate] {}
++(\figurangle:\figurlen) node (c1) [gate] {}
++(0:\figurlen) node (o1) [gate] {}
(c1) ++(0:-\figurlen) node (y1) [gate] {}
++(0:-\figurlen) node (y2) [gate] {}
+(\figurangle:-\figurlen) node (ot2) [gate] {};

\foreach \from/\to in {ot2/y2, y2/y1,
                       y1/c1, i1/y1, i1/x1, x1/c1, c1/o1, i1/y2} {
\path [draw, ->] (\from) -- (\to);
}

\foreach \where/\anchor/\text in
{ot2/south/$X_0$,
y2/north/$B^{(1)}$, i1/south/$I$, x1/south/$A^{(2)}$,
y1/north/$B^{(2)}$, o1/north/$O^{(3)}$, c1/north/$C^{(3)}$} {
\node [anchor=\anchor] at (\where) {\footnotesize \text};
}
\end{tikzpicture}
}

\def\figgawidth{2.8}
\def\figgaheight{1.7}
\def\figgayshift{0.15}

\def\figgamuxwidth{0.45}
\def\figgamuxheight{1}
\def\figgamuxindist{50}
\def\figgamuxtoout{0.25}
\def\figgaoutlen{0.25}
\def\figgarstlen{0.15}
\def\figgainlen{0.25}

\def\tikzfiguregate{
\begin{tikzpicture}[
gate/.style={draw, inner sep=0},
xscale=1,
yscale=0.4,
>=latex'
]

\coordinate (origin) at (0,0);
\coordinate (topright) at (\figgawidth,\figgaheight);

\path [draw] (origin) rectangle (topright);

\path [draw]
($(origin -| topright)!0.5!(topright) + (0,\figgayshift)$)
+(\figgaoutlen,0) node [anchor=west] {\small $v$} --
+(-\figgamuxtoout,0) node (mux) [draw, trapezium, anchor=top side,
shape border rotate=270, minimum width=\figgamuxheight*1cm, 
minimum height=\figgamuxwidth*1cm, trapezium stretches body] {};

\foreach \muxin/\muxincpt/\name/\cpt in {1/0/bool/$b$, -1/1/const/$I$} {
\path [draw] 
(mux.{180+\figgamuxindist*\muxin}) node [anchor=west, xshift=-1.5pt] {\footnotesize \muxincpt} --
++(-\csname figga\name tomux\endcsname, 0)
node (\name) [gate, anchor=east, minimum width=\csname figga\name width\endcsname*1cm,
minimum height=\csname figga\name height\endcsname*1cm] {\small \cpt};
}

\path [draw, densely dashed, thin]
(mux.north) -- (mux.north |- topright) --
++(0, \figgarstlen) node [anchor=south, inner sep=1.5pt] {\small\em reset};

\foreach \pos in {0.125, 0.375, 0.625, 0.875} {
\path [draw]
($(bool.north west)!\pos!(bool.south west)$) coordinate (tmp) --
(tmp -| origin) -- ++(-\figgainlen,0);
}
\end{tikzpicture}
}